\documentclass[]{article}

\usepackage[a4paper, total={6.5in, 9.5in}]{geometry}
\usepackage{blindtext}
\usepackage[T1]{fontenc}
\usepackage[utf8]{inputenc}
\usepackage{indentfirst}
\usepackage{lscape}

\usepackage{chngcntr}
\usepackage{amssymb}
\usepackage{amsmath} 
\usepackage[utf8]{inputenc}
\usepackage{multicol}
\usepackage{amssymb}
\newtheorem{theorem}{THEOREM}
\newtheorem{proof}{PROOF}
\usepackage{chngcntr}
\usepackage{graphicx}
\usepackage{mathtools}
\usepackage{sectsty}
\usepackage[utf8]{inputenc}
\usepackage[english]{babel}
\usepackage{moreverb}
\usepackage{setspace}
\usepackage{enumitem}
\usepackage{booktabs}
\usepackage{multirow}
\usepackage{flafter}

\counterwithout{equation}{section}

\begin{singlespace}
\title{A study on bribery networks with a focus on harassment bribery and ways to control corruption}
	\author{Chanchal Pramanik \thanks{Email id: cp629@cornell.edu; Tata-Cornell Institute has supported the research work. This paper is under review by the Journal of Public Economic Theory.}}
\date{Cornell University, Ithaca, New York}
\end{singlespace}

\fontsize{12}{12}\selectfont

\begin{document}

\maketitle

\begin{abstract}
The paper focuses on the bribery network emphasizing harassment bribery. A bribery network ends with the police officer whose utility from the bribe is positive and the approving officer in the network. The persistent nature of corruption is due to colluding behavior of the bribery networks. The probability of detection of bribery incidents will help in improving controlling corruption in society. The asymmetric form of punishment and award equivalent to the amount of punishment to the network can enhance the probability of detection of harassment bribery $(p_{h})$ and thus increasing the probability of detection of overall bribery $(p_{h} \in p)$.
\end{abstract}

\section{Introduction}

Corruption is part of our society for ages. It has many forms. Bribery is one form of corruption that enables an individual to get an advantage in unfair means. It should not be confused with harassment bribe (Basu, 2011) or extortion bribe (Lindgren, 1992). This type of bribe can be two types, harassment or extortion by threats or fear and seeking a corrupt payment by a public official because of the officer's ability to influence an official action (Lindgren, 1992). The harassment bribe under the color of office is studied in this paper, emphasizing on the police system in developing countries perspective, especially India. 

The effectiveness of the police system has a significant influence in shaping the behavior of society, and thus economic welfare. It is employed by the government to establish law and order and control corruption in our society. But it is unfortunate to get influenced by the police system to promote corruption in society. This type of malpractice is rampant in developing countries, like India. The country ranked 80 in the world and ranked 14 in the Asia Pacific region (Transparency International, 2019). The recent survey by the organization of sample size 2,000 in the period June-July 2020 reveals that the bribery rate in India is 39\% (Transparency International, 2020). In that survey, it is estimated that having contacted with police in the last 12 months, 42\% of Indians paid a bribe to a police officer in order to get assistance or to avoid a problem like passing a checkpoint or avoiding a fine or arrest. Similarly, 32\% Indians paid a bribe to a judge or court official. It is evident that the incidence of corruption in the judicial system is quite high in India with similar observations in other developing countries in the Asia Pacific region.

The paper also envisions developing a judicial system to limit the growth of the bribery network and explores how to nudge the law keepers, to be honest in delivering their services. The dishonest and selfish behavior of administrations diverts the society into a bad equilibrium, thus encouraging breeding grounds for new forms of corruption and injustice. This encourages erosion of trust in a society which degrades the environment for economic activities. Gradually, society under-perform in human development indices. The erosion of the moral behavior of a society affects the communities interconnected with the society. The modern-day digital world connected through social media amplifies the information flow and influence our actions and thus shape societies rapidly.

It is evident from the literature that bribery in the police system can be controlled in various ways and governments have raised multiple steps in controlling it. The recursive behavior in sharing of bribery amounts keeps the bribery cases confidential in the police system, thus it establishes a norm of giving a bribe in society. This paper analyzes the characteristics of the bribery network through Nash bargaining and game-theoretic analyses.

The model developed in the paper comprises of two sections. Firstly, the bribe size is evaluated by analysis at the network level. Secondly, the bribe size is determined for the police officer at the aggregate level in a specified time period in the network. The results of the analyses are examined under the comparative statics section. The conclusion section discusses the findings and future work.

\section{Literature Review}

The governments' effort to control crime comprises imposing fine, suspension from jobs, or imprisonment. The control of the crime in society is also related to the probability of identifying the crime. Becker (1968) concluded that "crime would not pay is an optimality condition and not an implication about the efficiency of the police or courts; indeed it holds for any level of efficiency, as longs as optimal values of $p$ (i.e., the overall probability that an offense is cleared by conviction) and $f$ (i.e., punishment for offense) appropriate to each level are chosen". This foundation work in the field of law and economics by Gary Becker construct the layout for further analysis on crime and punishment through the lens of economics. The crime, punishment, and probability of detection of crime give rise to the incidents of bribery to the law enforcement officers to avoid crime. The bargaining of bribe amount with a given probability of detection and size of fine can be analyzed through the setup of the bargaining problem (Nash, 1950). This paper analyzes the incidents of bribery and harassment bribery on the foundation of these two papers, Becker (1968) and Nash (1950).

Corruption is prevalent in developing countries and has a detrimental impact on the economic welfare of societies. Corruption as defined by Basu (2016), "can take many forms but, in the final analysis, it is a form of violation of the law, perpetrated either individually or in cahoots with state officers and enforcers of the law, as happens in cases of bribery". Basu (2016) further argued that law enforcement officials are susceptible to bribery, which leads to the recursive problem in bribery literature. The model capturing the recursive problem in the bribery incidents give birth to the bribery network, the characteristics of the network can be understood by developing a model using Nash bargaining setup (Basu, Bhattacharya, and Mishra, 1992). The analysis of the bribery models will reveal useful insights into the network and suggest ways to control the bribery incidents. The discussions of the issues in the literature are compliant with the author's experience on passport application where the inquiring police officer revealed the existence of the bribery network and optimal bribe size in 2007, India. It is an important issue to study to maintain law and order in society and encourage economic growth.

The growth of bribery incidents (Transparency International, 2019 and 2020) reveals the pervasive nature of corruption in society. The persistence of the pervasive corruption is due to low compliance in society, driving the society where corruption becomes a social norm (Mishra, 2006). This also refers to the existence of multiple equilibria, where it is more expensive to become honest in a corrupt environment (Basu, Bhattacharya, and Mishra, 1992). The ambiance of a corrupt environment in society gives rise to the harassment bribery where the bribe giver is forced to give a bribe. To control harassment bribery asymmetric punishment will empower the bribe givers to expose the incidents and will help in controlling corruption in society (Basu, 2011).

A more detailed study on asymmetric punishment for bribery reveals that asymmetric punishment eliminates bribery only when the cost of whistle blowing is cheap (Basu, Basu, and Cordella, 2016). The paper further reports that the bribe size increases as the expected penalties rise but bribery eliminates if the expected penalties are sufficiently high. The bribery incidents are detected by imposing elite police officers in the system placing an audit mechanism on a periodic basis (Basu, Bhattacharya, and Mishra, 1992). Given the measures, in place to control corruption, there is evidence of persistence of corruption, even in severest punishment (i.e., execution) environment as in China. This may be due to colluding behavior of the bribery network, leading to many competing bribe sizes and multiple equilibria with different levels of corruption (Bardhan, 2015).

\section{The Problem}

An ordinary citizen applied for a service (say, passport, driving license, etc.) to the government and eligible to receive the service. The citizen needs to go through a document verification process conducted by a police officer or a government official. The successful verification of the documents will issue the requested document from the government. But in a real-life situation, ordinary citizens are forced to pay a fee (i.e., a bribe or a harassment bribe) to the interacting government employee to pass the document verification process.

Let us define the problem based on the interaction between the police officer and the citizen in a story form, which will present the ambiance of the bribery network understanding of the crisis to find solutions for the problem. Suppose the citizen is a student in economics and waiting for his document verification process of passport by the local police authority. The student met a corrupt police officer and he demanded a bribe to approve the student's documents to issue the passport. The student declined to pay the bribe, though the police officer lowered the bribe amount. The police officer tried to persuade the student by telling him that the bribe amount needs to be shared with another officer(s) to clear the police verification report. The police officer added, without the bribe amount the police verification report cannot be cleared, and hence passport will not be issued. That information encouraged the student to study the literature around harassment bribery and design a theoretical analysis to help him decide the best strategy to combat the situation. 

The police officer tried to help the student by negotiating the bribe amount, but could not waive out the bribe amount. The police officer was forced to collect and share the bribe in the system, being part of the network and job. The phenomenon can be well understood by simple game-theoretic analysis. Let the approving police officer was a dishonest person who requires a bribe for his welfare and/or to accumulate money to keep the bribery network active. Thus, the subordinate police officer needs to collect a bribe, otherwise, his payoff is negative to clear the document verification process. The citizens need to pay the bribe to avoid harassment, which is costlier than the bribe amount. Thus, charging bribe and treating it as a fee in society is the only Nash equilibrium both from the sub-ordinate police officer's and citizen's perspectives.

The bribery activities are supported by a network of police officers in the order of their ranks. Thus, negotiation on bribery activities is with the citizen and the bribery network. Given the corrupted networking system it is challenging for the student or ordinary citizen to find a way to avoid giving a bribe to receive a passport from the government. This paper attempts to find the answers to this question, in particular.

\section{Notations}

Let a sub-ordinate police officer $(O_{1})$ charged a bribe $(B_{i})$ to a citizen $(C)$. The minimal bribe amount that can be offered is $(B_{i}')$. Let the payoff obtained on receiving the passport by the citizen is $P \in [0,P]$, whereas the police officer $(i)$ receives $B_{i} \in [0,\mathbb{B}_{i}]$ from bribe, if demanded. The expected total bribe amount in a certain period is $\mathbb{B}_{i}$. Given the bribe network exist, let the share of the bribe, $x \in [0,1]$, will be received by the interacting police officer $(O_{1})$ and $(1-x)$ of the bribe $B_{i}$ will be shared with another police officer in a sequence of the rank, and it follows.

The government's regulation mechanism to control the bribe network is done by imposing a fine $(F_{O})$ on the convicted officer and a fine $F_{C}$ on the convicted citizen if caught in the bribing process. The punishment is symmetric if $F_{O}=F_{C}$ and asymmetric is $F_{O}>F_{C}$. Let the probability of detection of bribe is $p \in [0,1]$ and harassment bribe is $p_{h} \in p$. Again, if the bribery incident is perceived and recognized as a harassment bribe the citizen is awarded an amount $f(B_{i}, F_{O})>F_{C}$. 

\section{The Model}

\subsection{The Bargaining setup}

The standard Nash bargaining setup has been used to measure the bribe size at network level and at aggregate level in a certain time period in the network for a police officer. The estimations are under following two subsections.

\subsubsection{Bargaining setup to determine bribe size by analysis at network level}

The Nash Bargaining setup in determining the size of the bribe is based on the utilities of the citizen and the participating police officers in the bribery network and incidences of harassment bribe and normal bribe, which includes harassment bribe.

The citizen's (C) utility function is given by,
\begin{equation}
\begin{split}
U_{C} (B_{i})&=P-B_{i}-pF_{C}+p_{h}f(B_{i},F_{O}) \\
&=P-B_{i}-pF_{C}+p_{h}(B_{i}+nF_{O})
\end{split}
\end{equation}

Assuming a linear award function for the citizen for harassment bribery incident, $f(B_{i},F_{O}) = B_{i}+nF_{O}$. The harassment bribe is a more serious crime committed by the police department. The detection of such incidents awards the bribe giver by returning the bribe amount and the fine collected from the officers involved in the bribery incidents. If the harassment bribe is claimed but not exchanged, then the citizen receives the passport $(P)$ and an award of $(nF_{O})$ collected from the police officers in the verified bribery network $(O_{1}, O_{2},..., O_{n})$, which is charged as fine to harass the citizen and support the bribery network. On the other hand, if the harassment bribe is exchanged the citizen receives the award and charged a fine $F_{C}$ (where, $F_{C}<f(B_{i},F_{O})$) as $p_{h} \in p$. In normal bribery incidents, the citizen is charged with a fine $F_{C}$.

The utility function of the police officer $(O_{1})$ interacting directly with the citizen $(C)$ is
\begin{equation}
U_{O_{1}} (B_{i})=xB_{i}-p(F_{O}+B_{i}))
\end{equation}

where, $B_{i}$ denotes the original amount shared with the first interacting sub-ordinate police officer $O_{i}$.

Let the police officers receive share $x$ of the received bribe to clear the file for the passport. The utility function of the police officers $(O_{j})$ interacting with the police officers in the bribery network to approve the citizen's $(C)$ passport verification report is as follows.
\begin{equation}
\begin{split}
U_{O_{2}} (B_{i})&=x(1-x)B_{i}-p[F_{O}+(1-x) B_{i}]\\
U_{O_{3}} (B_{i})&=x(1-x)^2B_{i}-p[F_{O}+(1-x)^2 B_{i}]\\
&...\\
U_{O_{n}} (B_{i})&=x(1-x)^{n-1}B_{i}-p[F_{O}+(1-x)^{n-1} B_{i}]
\end{split}
\end{equation}

The expected utility of the bribery network is derived based on von Morgenstern expected utility function, assuming $\dfrac{1}{n}$ as the probabilities associated with each utilities,
\begin{equation}
\begin{split}
U_{O}(B_{i})&=U_{O}[\dfrac{1}{n}\sum_{j=1}^{n}U_{O_{j}}(B_{i})] = \dfrac{1}{n}\sum_{j=1}^{n}U_{O_{j}}(B_{i}) \\
&=\dfrac{1}{n}[x.\dfrac{1}{x}B_{i}-p(nF_{O}+\dfrac{1}{x}B_{i})] \\
&=\dfrac{B_{i}}{n}-p(F_{O}+\dfrac{B_{i}}{nx})
\end{split}
\end{equation}
Since, $x\in(0,1)$, for large $n$, $1+ (1-x) + (1-x)^2 + ...+(1-x)^n=\dfrac{1}{x}$.

The comparison of utilities for a successful bribery network is summarized under Axiom 1.
\subsection*{Axiom 1:}
\begin{enumerate}
	\item The bribery network exist for $P\geq P^{'}$ such that $U_{C}(B_{i})>0$.
	\item Having the bribe $(B_{i})$ demanded the utility of individual police officers is less than the network utility from the bribe $(B_{i})$. thus individual officers need to share the bribe in the network, given the probability of detection of overall $(p)$ and harassment $(p_{h})$ bribes are non-zero. Explicitly, $U_{O_{j}}(B_{i}) \ngtr U_{O}(B_{i})$ as $x(1-x)^{j-1} \ngtr 1$, $x\in (0,1)$.
	\end{enumerate}

The bribe size is determined by the Nash Bargaining setup. Given the bribery network exists, a Nash bargaining solution exists if there is the gain in the transaction of bribes. Thus, the solution is given as follows. It is an encounter of differences of utilities (citizen and bribery network) and the probabilities of detection associated with the bribes in a bargaining situation\footnote{In the competing bargaining models the literature support the wider domain of applicability of the Nash bargaining set-up (Kalai and Smorodinsky, 1975; Anant, Mukherji, and Basu, 1990).}.

\begin{equation}
B_{i}^{*} \equiv \underset{B_{i}}{\operatorname{argmax}} [U_{C}(B_{i})-U_{O} (B_{i})][pB_{i} - p_{h}B_{i}]
\end{equation} 

At equilibrium the bribe size is as follows:
\begin{equation}
B_{i}^{*}=\dfrac{nx[P+p(F_{O}-F_{C})+p_{h}nF_{O}]}{2[nx(1-p_{h})+x-p]}
\end{equation}

The corresponding utilities by substituting the derived bribes size are as follows:
\begin{equation}
U_{C}(B_{i}^{*})=P+p_{h}F_{O}-pF_{C}-\dfrac{nx(1-p_{h})[P+p(F_{O}-F_{C})+p_{h}nF_{O}]}{2[nx(1-p_{h})+x-p]}
\end{equation}
\begin{equation}
U_{O}(B_{i}^{*})=\dfrac{(x-p)[P+p(F_{O}-F_{C})+p_{h}nF_{O}]}{2[nx(1-p_{h})+x-p]}-pF_{O}
\end{equation}

From equation (6), for $B_{i}^{*}>0$, we can draw insights for equation (8). For $x>p$ and $n>2$, $U_{O}(B_{i}^{*})>0$, as $p_{h}\in (0,1)$ and $x \in (0,1)$. Again, increase in the probabilities of bribe detection $p_{h}$ and $p$ negatively impact the utility of the network and individual officers in the network in bribery business, referring equation (6). Progressing in that argument as $p_{h} \in p$, and citizens contribute in increase of $p_{h}$ (i.e., increase in detection of harassment bribery incidents\footnote{One can report the incidents with the help of digital technologies, for example, http://www.ipaidabribe.com/}), the utility from bribe for the officers will decrease, having the penalty $(F_{O}, F_{C})$ and award setup $[f(B_{i}, F_{O})]$ in place. The individual utilities of the police officers are positive only when $\dfrac{B_{i}^{*}}{F_{O}}>\dfrac{p}{(1-x)^{n-1}(x-p)}$ i.e., $p>\dfrac{x}{2+x(n-1)}$, comparing (3) and (6). Hence, for having positive utility for individual police officers and the network, the relation is as follows,
\begin{equation}
x>p>\dfrac{x}{2+x(n-1)}
\end{equation}
In other words, the bribe size is directly proportional to the density of the bribery network, for a given share and the probability of detection. Thus, small bribe size relates to small bribery network.

Similarly, by comparing equations (1) and (9), the following relation can be deduced for $p_{h}$, considering asymmetric punishment i.e., $(F_{O}-F_{C})>0$ and $p_{h}\in p$.
\begin{equation}
\begin{split}
\dfrac{p}{n}&>p_{h}>\dfrac{B_{i}+pF_{C}}{B_{i}+nF_{O}} \\
\implies x>p&>p_{h}>\dfrac{x}{n[2+x(n-1)]}\dfrac{F_{C}}{F_{O}}>\dfrac{x}{2+x(n-1)}
\end{split}
\end{equation}
The relation (10) highlight the fact, $p_{h} \rightarrow p$ as $F_{C} \rightarrow F_{O}$.

Further from (5), it can be observed that $\dfrac{\delta B_{i}^{*}}{\delta F_{O}}>0$ and $\dfrac{\delta B_{i}^{*}}{\delta F_{C}}<0$. This implies that as the bribe size increases the penalty to the officer(s) is more, symmetrically the increase in the fine for the citizens will decrease the bribe size. The growth of bribe size also depends on probabilities, the share of bribe, and the number of police officers in the bribery network. Thus, it is important to study the pattern through comparative statics.

It is assumed, $p_{h} \in p$. Considering the particular case for the event of $B_{i}$, there can be three possibilities - harassment bribery transaction didn't take place but charged, harassment bribery transaction took place and the citizen paid the bribe, and general bribery took place where citizen took advantage by bribing the police officer. Thus, $p$ and $p_{h}$ are treated separately in this context.

Holding $p_{h}$ as constant, and allowing $p$ to vary,
\begin{equation}
\dfrac{\delta B_{i}^{*}}{\delta p}=\dfrac{nx}{2}[\dfrac{x[p_{h}(n-1)+1](F_{O}-F_{C})+(P+p_{h}nF_{O})}{[nx(1-p_{h})+x-p]^2}]\geq 0
\end{equation}

Holding overall probability of bribe detection as $p$ constant, and allowing $p_{h}$ to vary,
\begin{equation}
\dfrac{\delta B_{i}^{*}}{\delta p_{h}}=\dfrac{n^2x}{2}\dfrac{[xp(F_{O}-F_{C})+(nx+x-p)F_{O}+xP]}{[nx(1-p_{h})+x-p]^2} \geq 0
\end{equation}

\begin{equation}
\dfrac{\delta B_{i}^{*}}{\delta n}= \dfrac{x(x-p)[P+p(F_{O}-F_{C})+2p_{h}nF_{O}]+p_{h}n^2x^2F_{O}(1-p_{h})}{2[nx(1-p_{h})+x-p]^2} \geq 0
\end{equation}

\begin{equation}
\dfrac{\delta B_{i}^{*}}{\delta x}=\dfrac{-np}{2}\dfrac{[P+p(F_{O}-F_{C})+p_{h}nF_{O}]}{[nx(1-p_{h})+x-p]^2} \leq 0
\end{equation}

Considering, the penalty for bribe hurts more to the officer than to the citizen (i.e., $F_{O}>F_{C}$), asymmetric punishment, and assuming the utility from the bribe is positive for the police officers (i.e., $(x>p$).

The equations (11) and (12) suggest that bribe size increases with the probabilities of detection and asymmetric punishment $(F_{O}>F_{C})$. Figure 1 and figure 2 demonstrates the findings explicitly. This intuitively, implies that the bribe needs to be shared with other police officers to cover the crime. The intuition is more clear from equation (13), with an increasing number of police officers the bribe size need to be increased to cover the incident of crime. Thus keeping the bribery network active. This naturally, relates to the share of the bribe in the network. The increase in the share of the bribe relates to a decrease in the size of the bribe, equation (14). This relates to the fact the smaller the bribe size larger the share of the bribe of the individual officer, and the smaller the network (i.e., $n$). Again, the smaller bribe sizes are less detectable, as in equations (11) and (12).
\begin{figure}
	\includegraphics[width=\linewidth]{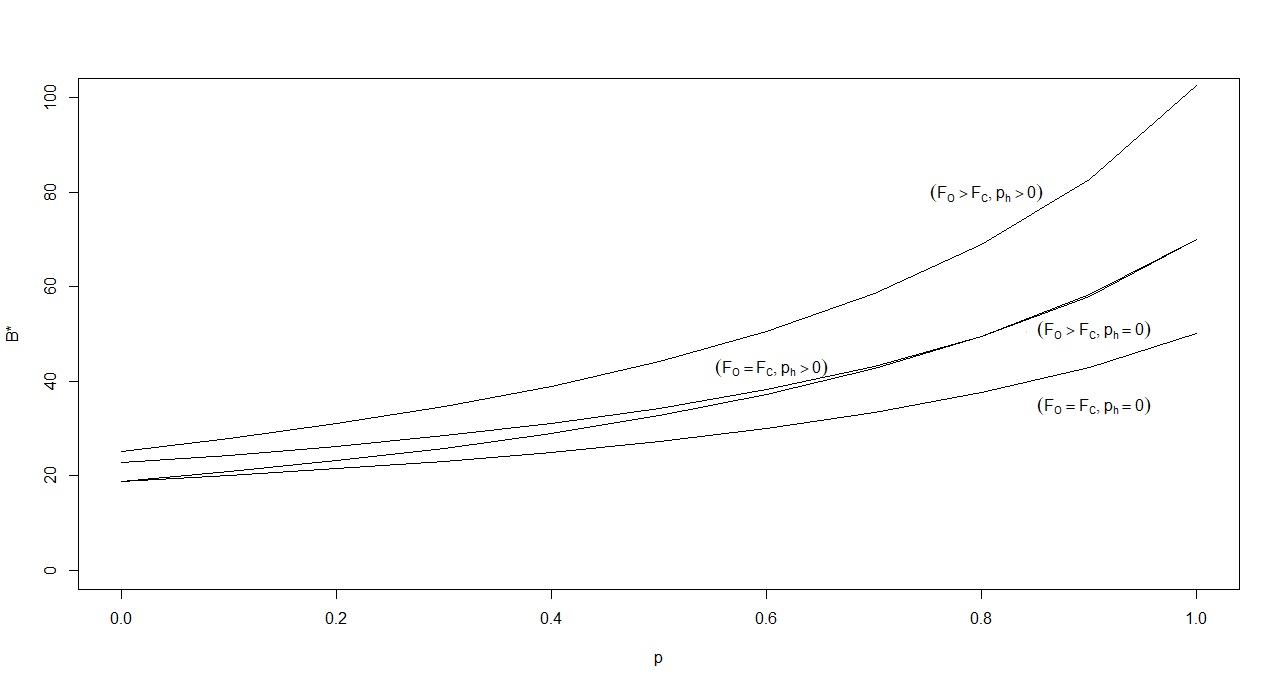}
	\caption{Equilibrium bribe size $(B_{i}^{*})$ as a function to detect probability of bribery $(p)$.}
	\label{fig:plot1}
\end{figure}
\begin{figure}
	\includegraphics[width=\linewidth]{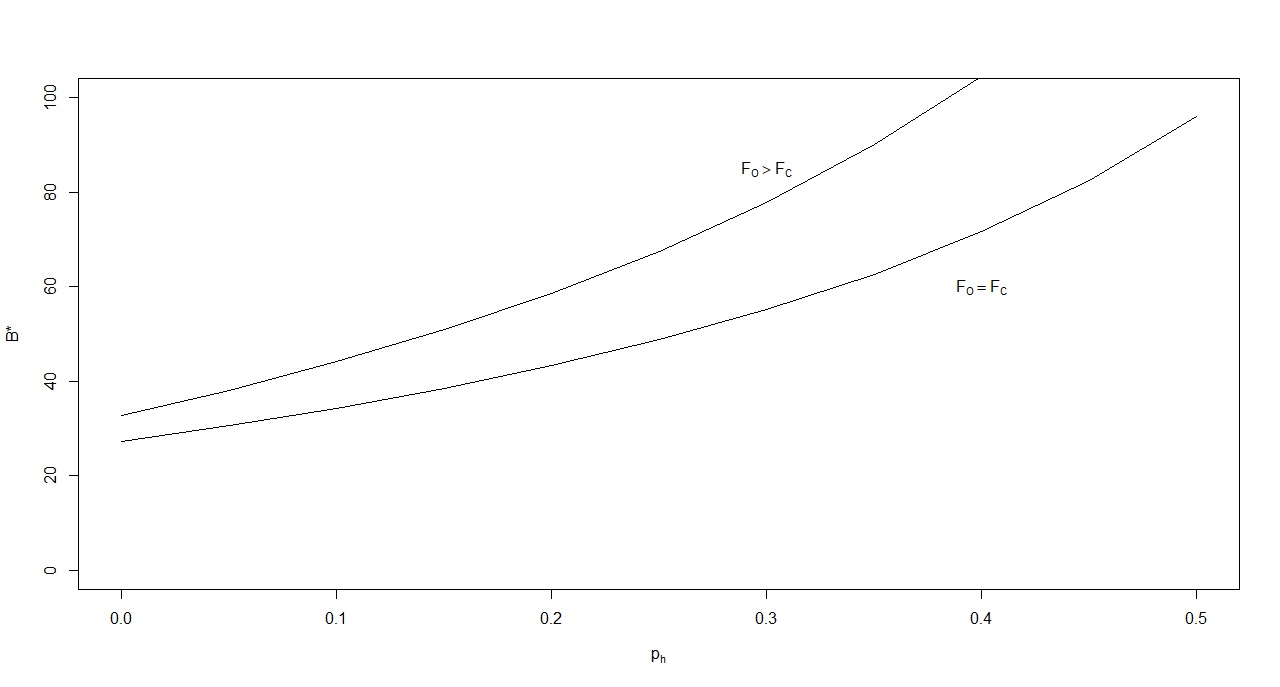}
	\caption{Equilibrium bribe size $(B_{i}^{*})$ as a function to detect probability of harassment $(p_{h})$ holding overall bribery $(p)$ constant at $0.5$.}
	\label{fig:plot2}
\end{figure}

From this analysis, it can be inferred that the bribe network expands and contracts relative to the bribe size. And the share of bribe is inversely proportional to the size of the bribe, considering the equal share of bribe $(x)$ in the network. The findings can be summarized under axiom 2.

\subsection*{Axiom 2:} 
\begin{enumerate}
	\item Considering the bribery network is active, $U_{O}(B_{i})>0, \; U_{C}(B_{i})>0$, that is, $x>p$ and there exist $P\geq P^{'}$ such that the bribe size is strictly increasing with $F_{O}$ and strictly decreasing with $F_{C}$.
\begin{enumerate}
	\item For the probabilities, $p_{h} \in p$, the bribe size is weakly increasing with a probability of detection of bribery $p$ having the probability of detection of harassment bribery $p_{h}$ as constant. Similarly, the bribe size is weakly increasing with a probability of detection of harassment bribery $p_{h}$ having the probability of detection of overall bribery $p$ as constant.
	\item The bribe size is weakly increasing with the number of police officers in-network $n$ and weakly decreasing with a share of bribe size $x$. 
\end{enumerate} 
\item There will be no bribery incidents if $p>x$ and $P<P^{'}$.
\end{enumerate}

The analysis of the bribery network can be summarized under following theorems.

\begin{theorem}
The bribe network ends with the police officer whose utility is positive from the bribe and the approving officer in the network.
\end{theorem} 

\begin{proof}
Let the bribes are shared in the hierarchy of the ranks of the police officers at a constant share of bribe $(x)$. That is, the bribe amounts become increasingly small till the individual utility is positive ($U_{O_{n}} (B_{i})>0$, equation 3). The bribe will not be shared further as the $U_{O_{n+1}} (B_{i})<0$. The $(n+1)$th officer will not accept the bribe as his utility from the bribe is less than zero and will not approve the file for further processing as the bribe is not received. Thus, the $n$th police officer must be the approving officer in the network for the passport application to be completed and having the bribe received from the citizen. The $n$th police officer is one of the nodal officers in the network, his aggregate collection of bribes must be greater than his subordinates, for a given share of bribe $(x)$ in the network.
\end{proof}

\begin{theorem}
	The bribe networks are colluding.
\end{theorem}

\begin{proof}
From theorem 1, it can be inferred that the (n+1)th police officer is in higher or equivalent rank in the administration and has the power to approve the service request similar to (n)th police officer. Considering the (n+1)th police officer is aware of bribery incident cannot charge bribe as his utility from the bribe $U_{O_{n+1}} (B_{i})<0$. Again, having the bribery network is elastic (Axiom 2), the (n+1)th officer will support the bribery network, to maximize his average utility from the network. Thus, the (n)th police officer will not share the bribe further, and (n+1)th police officer will not accept the bribe is the Nash equilibrium for the bribe $B_{i}$. Hence, considering the flow of the information for the bribe $B_{i}$ gets interrupted at $(n)$th police officer. This implies that the probability of detection of bribe $B_{i}$ remains low in the police department at a higher level. Again, the bribery network being long the probability of detection remains low, equation (11). But if the (n+1)th officer is an honest or elite police officer (i.e., immune to bribery), the bribery network fails (Basu, Bhattacharya, and Mishra, 1992). Thus, colluding nature of the bribery networks, sensitive to the probability of detection and individual utility, will keep the networks active.
\end{proof}

\subsubsection{Bargaining setup to determine bribe size for the police officer $(O_{n})$ at aggregate level in a specified time period in the network}
Taking the argument forward from Theorem 2. Suppose there is an elite police officer appointed in the police system who is immune to bribe exchanges with his subordinates. Now let, the bribe size depends on the income difference $(D_{r})$ with the elite police officer in the chain (say, the (n+1)th police officer). Thus, the total bribe size in a time period (say, monthly, quarterly, or semi-monthly) from own collection of bribes is $\mathbb{B}_{i}$ and bribes collected from other police officers is $\mathbb{B}_{k}$.
\begin{equation}\begin{split}
D_{r} &\geq x\sum_{i=1}^{m_{1}}B_{i}+\sum_{i \ne k=1}^{m_{2}}x(1-x)^{r_{k}-1}B_{k} \\
\implies D_{r} &\geq x\sum_{i=1}^{m_{1}}B_{i}+ x \sum_{i \ne k=1}^{m_{2}}(1-x)^{r_{k}-1}B_{k} \\
\implies \dfrac{1}{x} D_{r} & \geq \mathbb{B}_{i} + \mathbb{B}_{k}
\end{split}
\end{equation}
where, $\mathbb{B}_{i}=\sum_{i=1}^{m_{1}}B_{i}$ and $\mathbb{B}_{k}=\sum_{i \ne k=1}^{m_{2}}(1-x)^{r_{k}-1}B_{k}=\sum_{i \ne k=1}^{m_{2}}[1-x(r_{k}-1)]B_{k}$ for $x\leq 1$
Given the bribery network exists, a Nash bargaining solution exists if there is gain in taking and giving bribe. Thus, if the solution exists it is given as the following, where $\mathbb{B}_{i}'$ and $\mathbb{B}_{k}'$ are the agreed bribe amounts.
\begin{equation}
\mathbb{B}_{i}'^{*} \equiv \underset{\mathbb{B}_{i}'}{\operatorname{argmax}} [\dfrac{1}{x} D_{r}-\mathbb{B}_{i}'][\mathbb{B}_{i}' - p \phi (\mathbb{B}_{i}',\mathbb{B}_{k}')] = \phi (\mathbb{B}_{i},\mathbb{B}_{k})
\end{equation}
Considering a linear relation in $ \phi (\mathbb{B}_{i}',\mathbb{B}_{k}')$ such that,
\begin{equation}
\phi (\mathbb{B}_{i}',\mathbb{B}_{k}') = \mathbb{B}_{i}' + \mathbb{B}_{k}'
\end{equation}
At equilibrium, considering $\dfrac{1}{x} D_{r} = \mathbb{B}_{i} + \mathbb{B}_{k}$ and equation (17)
\begin{equation}
\mathbb{B}_{i}' = \dfrac{\mathbb{B}_{k}'[1-p(1-x)]}{(1-2x)(1-p)}
\end{equation}
From (15) and (18), considering average size of bribe as $B_{i}^*$, the average minimum number of bribes to be collected by the sub-ordinate officers to satisfy the income difference with elite police officer is $m_{1}$.
\begin{equation}
\begin{split}
m_{1}&=\dfrac{1-p(1-x)}{p+(1-x)(2-3p)}\dfrac{D_{r}}{B_{i}^{*}} \\
&= \dfrac{[1-p(1-x)]D_{r}}{p+(1-x)(2-3p)} \dfrac{2[nx(1-p_{h})+x-p]}{nx[P+p(F_{O}-F_{C})+p_{h}nF_{O}]} \; [from \; (6)]
\end{split}
\end{equation}
The expression $m_{1}$ also represents the number of bribes to be collected by an individual sub-ordinate police officers at equilibrium targeting to meet an income level (say, income difference with an elite police officer).

\section{Comparative Statics}
The expression (19) presents important findings that can be studied under Comparative Statics of this paper, given $B_{i}^*>0$. The number of bribes to be collected by individual sub-ordinate police officer minimizes to zero (since $m_{1}\geq 0$) when the income difference with elite police officer $D_{r}=0$. The police officers will have a personal incentive to collect bribes as long as $D_{r}>0$. Considering the elite police officer is corrupted the income difference will remain positive with an active bribery network. If the elite police officer is not corrupted and ignorant of bribery incidents, the bribe collection will be active till the income difference is met. Thus, to maximize the income of the police officers in the bribery network it is desirable by the bribery network to corrupt the elite police officers. This tendency of the law enforcement authorities will drive the society in low-level equilibrium where bribing is a social norm.

The frequency of bribe collection decreases with the increase of probability of detection $(p)$ and harassment bribery $(p_{h})$, keeping the other probability as constant. The plots reveal the trend for certain parameter values, in figure 3 and figure 4. The trends are different with respect to the bribe size $B_{i}$, figure 1 and figure 2. This implies the frequency of bribe collection is indirectly proportional to the bribe size. This also reflects a similar relationship with a share of bribe and bribe size, equation 14. The findings from the plots also reflect that asymmetric punishment has a negative impact on the frequency of bribery incidents relative to symmetric punishment. Further, asymmetric punishment with an increase in the probability of detection of harassment bribery minimizes the incidents of bribery, figure 3. The argument can be explicitly followed by referring to figure 4. Thus, the combination of asymmetric punishment and initiatives to increase the probability of harassment will help to break the colluding bribery networks.

From the previous argument, it is followed that the presence of elite police officers may encourage bribery incidents. But the introduction of asymmetric punishment and increase of probability of harassment bribery (that is, reporting to elite police officers) will encourage the functioning of the law enforcement machinery to control incidents in society. This may result in an increase in bribe size, which will increase the probability of detection of bribery incidents, figure 1 and figure 2. Thus, society will tend towards the good equilibrium.

The incidence of bribery cases $(m_{1})$, can be a measure of trust in society. The surveillance on the incidence of bribery cases from the audit system of the police system and encouraging citizens to contribute can help in controlling the bribery incidents in society. The modern-day availability of digital technologies can help to report the incidence of bribery cases, thus improving the trust measure in society.

\begin{figure}
	\includegraphics[width=\linewidth]{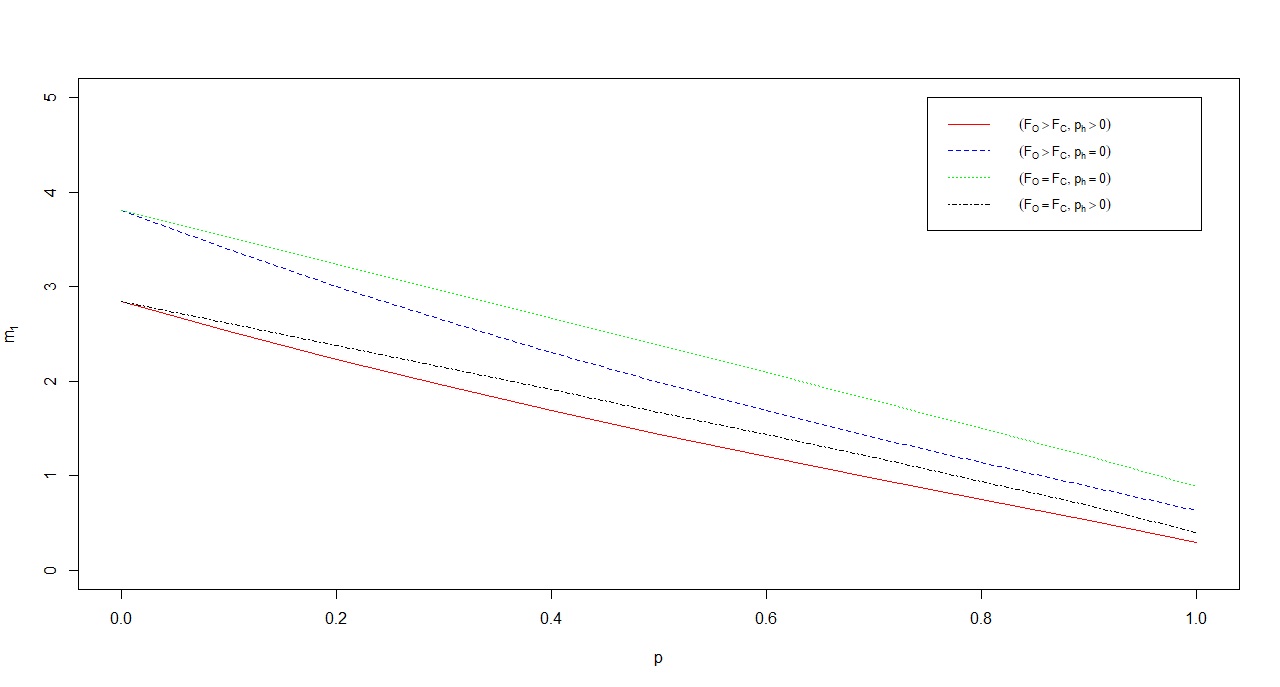}
	\caption{Equilibrium bribe number $(m_{1})$ as a function to probability of detection of bribe.}
	\label{fig:plot3}
\end{figure}

\begin{figure}
	\includegraphics[width=\linewidth]{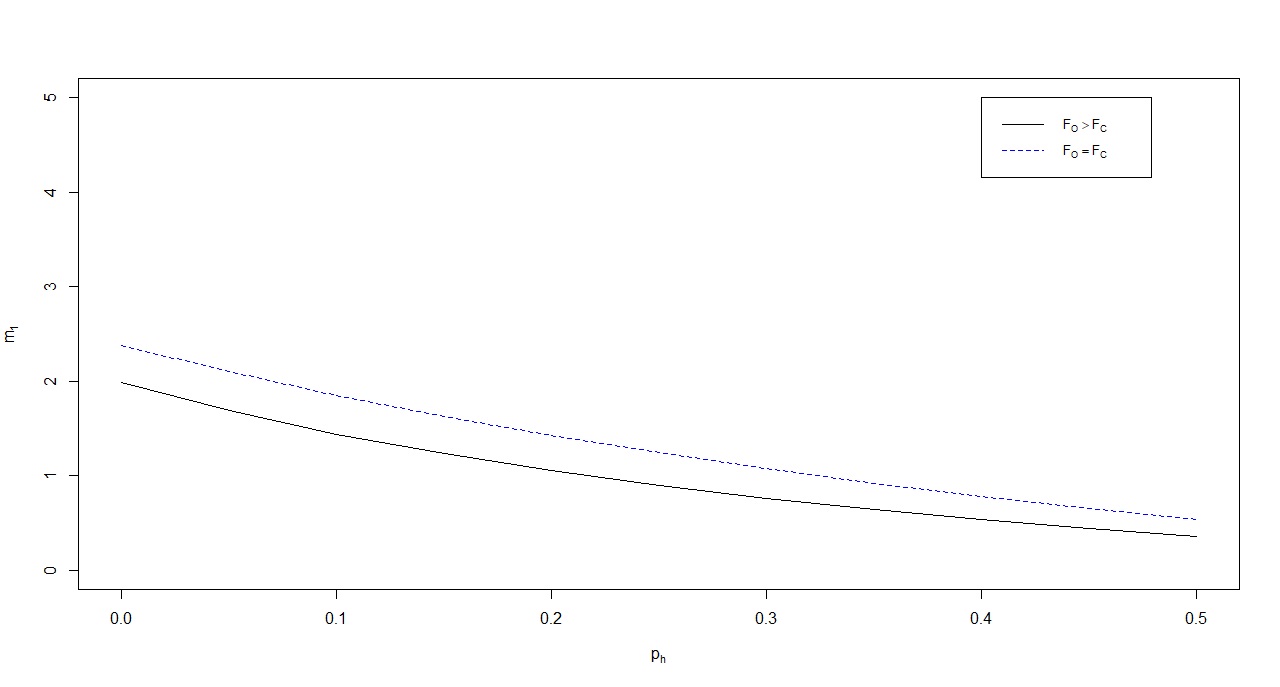}
	\caption{Equilibrium bribe number $(m_{1})$ as a function to detect probability of harassment $(p_{h})$ holding overall bribery $(p)$ constant at $0.5$.}
	\label{fig:plot3}
\end{figure}

\pagebreak

\section{Conclusion}

Corruption is inevitable in society unless there is a realization of its ill effects on society's welfare. One of the harmful consequences of corruption is the erosion of trust, negatively impacting economic welfare. Harassment bribery is an example of such an act. In a very corrupt society where bribing becomes a social norm, bribe-givers are also bribe-takers in their respective businesses. Society gets diverted towards a direction where there exist multiple equilibria with different levels of corruption. The mistrust within the society is harmful to the local business environment and has a negative impact on investments in the society. Thus, giving rise to poverty and other issues related to development. Having the importance of the problem the paper is focused on bribery with emphasis on harassment bribery to find ways to control corruption in society.

The analysis in this paper reflected the bribing behavior at an individual police officer and network level. The first part of the analysis in this paper concentrated on the bribe size concerning the utilities of the citizen, individual police officer, and the network. It reveals two important conclusions. Firstly, the bribe network ends with the police officer whose utility is greater than zero and the approving officer in the network. Secondly, the bribe networks are colluding. The bribery incidents survive only there is a strong network of support. Thus punishment needs to be directed towards the network to combat these types of corruption. The asymmetric form of corruption and reward to equivalent to the fine imposed to the network (for example, $(B_{i}+nF_{O})$, in this analysis) will encourage whistle-blowers to expose the bribery incidents.

The second part of the analysis is focused on the frequency of bribery incidents. Considering the individual police officer is driven by a target to satisfy a certain income level (say, income difference with an elite police officer in the network), how many bribes one needs to collect to meet the threshold in a certain time period? The analysis reveals that bribery incidents decrease with a decrease of income difference. Thus, an increase in income differences can be an equilibrium point, where police officers can increase their income from bribes. On the positive aspect, the bribery incidents decrease with the increase in probabilities of detection (including harassment bribery), which will increase the bribe size to meet the income differences. Again, an increase in bribe size increases the probability of detection, from the first part of the analysis. Thus, efforts to decrease bribery incidents will control corruption in society.

Finally, the paper emphasizes endogenizing the probability of detection of bribery incidents $(p)$ by including the probability of harassment bribe $(p_{h} \in p)$. The arguments to control corruption can be further strengthened by identifying major players in the bribery processes and analyze their best replies under the framework of closed under rational behavior as defined by Basu and Weibull (1991). The government or political will can be one of the major players to divert the society towards a good equilibrium. The idea is left for future research.

\end{document}